\documentclass[reqno,12pt,letterpaper]{amsart}

\usepackage[usenames,dvipsnames]{color}
\usepackage[colorlinks=true,linkcolor=Red,citecolor=Green]{hyperref}
\usepackage{amsmath,amssymb,amsthm,graphicx,mathrsfs,url,bbm,mathtools,cases,enumitem,tikz-cd}
\usepackage{amsxtra}
\usepackage[toc,page]{appendix}
\usepackage{wasysym} 
\usepackage{graphicx}
\usepackage{subcaption}
\usepackage{comment}
\def\arXiv#1{\href{http://arxiv.org/abs/#1}{arXiv:#1}}

\usepackage{soul}

\usepackage{mathtools}

\def\?[#1]{\textbf{[#1]}\marginpar{\Large{\textbf{??}}}}

\def\smallsection#1{\smallskip\noindent\textbf{#1}.}
\let\epsilon=\varepsilon 

\usepackage{CJKutf8}
\newcommand{\RR}{{\mathbb R}}
\newcommand{\NN}{{\mathbb N}}
\newcommand{\CC}{{\mathbb C}}
\newcommand{\TT}{{\mathbb T}}
\newcommand{\ZZ}{{\mathbb Z}}

\newtheorem{theo}{Theorem}
\newtheorem{prop}{Proposition}[section]	
\newtheorem{defi}[prop]{Definition}
\newtheorem{ex}{Example}

\newtheorem{assumption}{Assumption}

\newtheorem{corr}[prop]{Corollary}
\newtheorem{rem}[prop]{Remark}

\numberwithin{equation}{section}

\DeclareMathOperator{\Spec}{Spec}

\DeclareMathOperator{\loc}{loc}

\DeclareMathOperator{\rank}{rank}

\let\Re=\Real

\DeclareMathOperator{\tr}{tr}

\def\indic{\operatorname{1\hskip-2.75pt\relax l}}
\usepackage{scalerel}

\newcommand\reallywidehat[1]{\arraycolsep=0pt\relax%
\begin{array}{c}
\stretchto{
  \scaleto{
    \scalerel*[\widthof{\ensuremath{#1}}]{\kern-.5pt\bigwedge\kern-.5pt}
    {\rule[-\textheight/2]{1ex}{\textheight}} 
  }{\textheight} %
}{0.5ex}\\           
#1\\                 
\rule{-1ex}{0ex}
\end{array}
}
\author[S.\,Becker]{Simon Becker}
\address[Simon Becker]{ETH Zurich, 
Institute for Mathematical Research, 
R\"amistrasse 101, 8092 Zurich, 
Switzerland}
\email{simon.becker@math.ethz.ch}

\author[Z.\,Tao]{Zhongkai Tao}
\address[Zhongkai Tao]{Department of Mathematics, University of California, Berkeley, CA 94720, USA}
\email{ztao@math.berkeley.edu}

\author[M.\,Yang]{Mengxuan Yang}
\address[Mengxuan Yang]{Department of Mathematics, University of California, Berkeley, CA 94720, USA}
\email{mxyang@math.berkeley.edu}

\title[Fragile topology]{Fragile topology on solid grounds: a mathematical perspective}

\begin{document}
\begin{abstract}
This paper provides a mathematical perspective on \emph{fragile topology} phenomena in condensed matter physics. In dimension $d \le 3$, vanishing Chern classes of bundles of Bloch eigenfunctions characterize operators with exponentially localized Wannier functions (these functions form convenient bases of spectrally determined subspaces of $L^2$).  However, for systems with additional symmetries, such as the $C_{2}T$ (space-time reversal) or the $PT$ (parity-time) symmetry, a set of exponentially localized Wannier functions compatible with such symmetry may not exist. We show that for rank 2 Bloch bundles with such symmetry, non-trivial Euler classes are obstructions to constructing exponentially localized compatible Wannier functions. We also show that this obstruction can be lifted by adding additional Bloch bundles with the symmetry, even though the Stiefel--Whitney class of the total bundle is non-trivial. This allows a construction of exponentially localized Wannier functions compatible with the symmetry and that is referred to as {\em topological fragility}.
\end{abstract}
\maketitle

\section{Introduction}
In condensed matter physics it is commonly accepted that topological features of Bloch bundles, such as the Hall conductance and the localization of Wannier functions, are robust under topologically trivial perturbations of the system. However, in the presence of certain symmetries, the so-called fragile topology discovered by Po--Watanabe--Vishwanath \cite{PWV18} challenges this narrative. Consider a periodic Hamiltonian invariant under the symmetry 
\begin{equation}
    \label{eq:I}
    Iu(x) =\overline{u(-x)}.
\end{equation} 
This symmetry induces a real sub-bundle $\mathcal{E}_{\RR}$ of the Bloch bundle $\mathcal E_{\CC}$ (cf.~equation \eqref{eq:bb}) of the periodic Hamiltonian. If ${\mathcal{E}_{\RR}}$ is of rank two with the Euler class $e_1({\mathcal{E}_{\RR}}) \neq 0$, there is no exponentially localized Wannier basis compatible with the $I$-symmetry (cf.~equation~\eqref{eq:compatible}).  
This topological phenomenon is considered {\em fragile}  because adding topologically trivial bands allows for the existence of an exponentially localized Wannier basis compatible with the $I$-symmetry. This topological feature has received attention because it prominently appears in twisted bilayer graphene \cite{LZSV19} and its flat bands in relation to strongly correlated electron phenomena \cite{seb,PSBH19}.

In this article, we give a mathematical characterization of fragile topology by proving a more general result stating that any real Bloch bundle $\mathcal{E}_{\RR}$ over $\TT^2$ or $\TT^3$ with rank $\geq 3$ induced by the $I$-symmetry admits an exponentially localized Wannier basis compatible with the $I$-symmetry -- see Theorem \ref{theo:main2}. In particular, increasing the rank of any rank 2 Bloch bundle destroys the topological obstruction to exponentially localized compatible Wannier bases. Our results also extend previous results by Po--Watanabe--Vishwanath \cite{PWV18} and Ahn--Park--Yang \cite{apy19} from two dimensional to three dimensional periodic systems. 
The nonexistence result in the rank 2 case is presented in Theorem \ref{theo:main1}.

We start by introducing some main concepts. 
\begin{assumption}
\label{ass:proj}
Let $\Gamma$ be a lattice in $\mathbb R^d$ and $\mathcal P:=(P(k))_{k \in \mathbb R^d}$ be a family of orthogonal projections with finite constant rank $r$ acting on the Hilbert space $\mathcal H = L^2(\mathbb R^d / \Gamma)$ that depends real analytically on the parameter $k \in \mathbb R^d$ and satisfies for all $k \in \mathbb R^d$ and $\gamma \in \Gamma^*$
\begin{equation}
\tag{A1}
\label{eq:t-equivariance}
P(k+\gamma) =\tau(\gamma)^{-1}P(k)\tau(\gamma), \quad \tau \in C^{\omega}( \mathbb R^d , U(\mathcal H) )\text{ a unitary operator}. 
\end{equation}
\end{assumption}
The most prominent examples of $P(k)$ are spectral projections to $r$ bands gapped from the rest of the bands for periodic Hamiltonians and $\tau(\gamma):=e^{-i\langle \gamma,x \rangle}$. 
A family of projections in Assumption \ref{ass:proj} induces a vector bundle via the equivalence relation
\[ (k,\varphi) \sim (k',\varphi') \in \RR^d\times\operatorname{ran}P(\cdot)  \Longleftrightarrow (k',\varphi')=(k+\gamma,\tau(\gamma)\varphi), \  \gamma \in \Gamma^*,\]
with the total space 
\begin{equation}
    \label{eq:bb}
    E:=\{ (k,\varphi) \in \mathbb R^d \times \operatorname{ran}(P(k))\} /{\sim}
\end{equation}
and the base space $B:=\mathbb R^d/\Gamma^*.$ Thus the projection map $\pi: E\to B$ induces a complex vector bundle $\mathcal E_{\CC}$ (cf.~Definition \ref{def:vb}), which is called the \textbf{Bloch bundle}. 

\begin{defi}
\label{def:wb}
A \emph{Wannier basis} of the Bloch bundle $\mathcal{E}_{\CC}$ is a set of orthonormal functions
\begin{equation}
    \label{eq:wb}
     \mathcal W := \bigcup_{a=1}^r \bigcup_{\gamma\in\Gamma} \big\{ \psi_a ( x - \gamma ) \big\} 
\end{equation}
such that $\mathcal W$ is an orthonormal basis of the range of the projection $\mathcal{P}$ with rank $r$.  
\end{defi}
In quantum mechanics, the existence of exponentially decaying Wannier basis is of particular interest, as this allows tight-binding constructions for the corresponding Hamiltonians. 
It is well-known that the existence of exponentially decaying Wannier functions is equivalent to the topological triviality of the Bloch bundle \cite{ne,P07,MPPT18}. 

In contrast, for Bloch bundles satisfying the $I$-symmetry in \eqref{eq:I}, we show that it may allow the existence of exponentially decaying Wannier functions compatible with the symmetry, even when the underlying symmetry induced \emph{real} Bloch bundle is topologically non-trivial. We introduce  
\begin{assumption}
\label{ass:I}
The projections $P(k)$ in Assumption \ref{ass:proj} satisfy 
\begin{equation}
\tag{A2}
    IP(k)=P(k)I.
\end{equation}
\end{assumption}
This symmetry is also well-known as $\mathcal{C}_{2}\mathcal{T}$ or $\mathcal{PT}$-symmetry in quantum mechanics. In particular, it holds for a large class of Hamiltonians, e.g.,~Bloch transformed Schr\"odinger operators $H_k= (-i\nabla-k)^2 + V$ with $V(-x)=\overline{V(x)}$ and the Bistritzer--MacDonald Hamiltonian of twisted bilayer graphene \cite{BM11}. For Bloch bundles with the symmetry, it is natural to study the existence of exponentially localized Wannier functions invariant under the $I$-symmetry. To this end, we introduce
\begin{defi}
\label{def:compatible}
    We say a Wannier basis $\mathcal{W}$ is \emph{compatible with the $I$-symmetry} if
    \begin{equation}
    \label{eq:compatible}
        I \mathcal W = \mathcal W.
    \end{equation}
\end{defi}
The $I$-symmetry acts as a real vector bundle homomorphism (i.e., each fiber is preserved). Since $I^2=1$, it has eigenvalue $1$ and $-1$ on each fiber. Hence, the $I$-symmetry naturally induces a real subbundle of the Bloch bundle
$$ \mathcal{E}_{\RR}:=\{(k,v)\in \mathcal{E}_{\mathbb{C}}: I v=v\},$$ 
which is
given by the eigenspace of $1$. This real subbundle $\mathcal{E}_{\RR}$ is closely related to the existence of exponentially decaying Wannier basis compatible with the $I$-symmetry: 

\begin{theo}[Wannier obstruction]
\label{theo:main1}
Let $\mathcal E_{\mathbb R}$ be a rank $2$ real orientable subbundle of the Bloch bundle induced by the $I$-symmetry with base dimension $\le 3$. The following are equivalent:
\begin{enumerate}
    \item \label{part1} The Euler class of $\mathcal{E}_{\RR}$ vanishes.
    \item \label{part2} $\mathcal{E}_{\RR}$ admits an exponentially localized Wannier basis compatible with the $I$-symmetry, i.e., for some $\beta>0$, we have
\begin{equation}
    \sum_{a=1}^2\int_{\RR^d} e^{2\beta|x|}|\psi_a(x)|^2 dx<+\infty.
\end{equation}
\item \label{part3} There exists a Wannier basis compatible with the $I$-symmetry satisfying the weaker decay condition
    \begin{equation}\label{eq:main1-2}
      \sum_{a=1}^2\int_{\RR^d} |x|^2|\psi_a(x)|^2 dx<+\infty.  
    \end{equation}
\end{enumerate} 

\end{theo}

\begin{rem}\label{rem:main1}
In contrast to Theorem \ref{theo:main1}, when the Euler class does not vanish, there exists a Wannier basis $\{ \psi_a ( x - \gamma ): \gamma\in\Gamma \}_{a=1,2}$ compatible with the $I$-symmetry satisfying
    \begin{equation*}
        |\psi_1(x)|+|\psi_2(x)|\leq C\left\{\begin{array}{ll}
          |x|^{-2},  &  d=2,\\
           |x|^{-7/3},  & d=3.
        \end{array}\right. 
    \end{equation*}
This is because a rank $2$ oriented real vector bundle can be identified with a complex line bundle. Hence, we can apply the result of \cite[Theorem~1,2]{bty} directly to obtain the corresponding Wannier decay. 
\end{rem}

In contrast, if the rank of $\mathcal E_{\mathbb R}$ is different from 2, we show 

\begin{theo}[Fragile topology]
\label{theo:main2}
Let $\mathcal E_{\mathbb R}$ be a rank $r$ real subbundle of the Bloch bundle induced by the $I$-symmetry with base dimension $\le 3$. When $r \neq 2$, the bundle $\mathcal E_{\RR}$ always admits an exponentially localized Wannier basis compatible with the $I$-symmetry.
\end{theo}

The proofs of the theorems are presented in Section \ref{sec:Frag_topo}. Thus, an obstruction of constructing exponentially localized Wannier functions compatible with the $I$-symmetry may exist only for vector bundles of rank $2$. As a direct consequence, consider the rank $2$ real subbundle $\mathcal E_{\mathbb R}$ without exponentially decaying Wannier basis; by adding a real line bundle $L$, the new bundle $\mathcal E_{\mathbb R} \oplus L$ automatically admits an exponentially localized Wannier basis compatible with the $I$-symmetry. This is the phenomenon of fragile topology in condensed matter physics -- see \cite{apy19,PWV18,bbs1,bbs2}.


\smallsection{Related works} 
Different invariance conditions \eqref{eq:conditions} for the Bloch and Wannier functions have been considered in \cite{FMP16a,FMP16b}. The observed effects are fairly different, and fragile topology does not appear in these settings. To compare our framework with the setting studied by \cite{FMP16a,FMP16b}, we can compare the $C_{2z}T$-symmetry $I_{C_{2z}T}u(x)=\overline{u(-x)},$ with the time-reversal symmetry $I_Tu(x)=\overline{u(x)}$ and notice that they act differently on Bloch functions. Indeed, while they are both bosonic in the terminology of the aforementioned works as they square to the identity, we have $[I_{C_{2z}T},(-i\nabla-k)^2] =0$ whereas $I_{T}(-i\nabla-k)^2 =(-i\nabla+k)^2 I_T$ for the time reversal symmetry. That is, time-reversal symmetry maps $k$ to $-k$, while we consider symmetries that leave $k$ unchanged; see Proposition  \ref{def:WB-c}. See also \cite{klw} for an account of the relation between the $\ZZ_2$ topological invariants (cf.~\cite{km,fk,fkm}) and the Stiefel--Whitney class in the setting of \cite{FMP16b}, and \cite{ahn} for classification of topological phases in bands with the $I$-symmetry using Stiefel--Whitney class from a physics perspective. The fragile topology phenomenon in various models of twisted bilayer graphene has been widely studied recently in condensed matter physics. See \cite{song19,seb,LZSV19,lu21,PSBH19} and the references therein. Our result also generalizes \cite[Proposition~5.3]{bmcone} on Wannier basis in twisted bilayer graphene. 

For mathematical results on periodic Hamiltonians, we refer the readers to \cite{ku3,ku2} and the references therein. The classical question on the existence of exponentially localized Wannier functions does not involve any symmetries and can be fully understood in terms of the Chern number \cite{ne,P07,br,MPPT18}. In particular, for periodic systems with non-trivial Chern number, even though it is not possible to construct analytic Bloch frames that lead to exponentially localized Wannier basis, it has been shown that one can construct so-called Parseval frames with redundancies (see \cite{ku,ak,cmm}) by embedding them into trivial Bloch bundles with higher ranks. This construction yields a set of exponentially localized Wannier functions that are \emph{linearly dependent}. Orthogonal to this direction of constructing linearly dependent Wannier functions using Parseval frames, the work \cite{bty} studies the optimal decay rate of Wannier basis from the global Bloch frame with singularities for Bloch bundles with non-trivial Chern classes. Recently, results on Wannier basis localization have been extended to non-periodic systems \cite{LSW,MMP,LS24,RP24}.

\smallsection{Notations and conventions} Let $\mathcal H$ be a separable Hilbert space. Let $\Gamma = \sum_{i=1}^d \ZZ e_i \subset \mathbb R^d,$ for linearly independent $e_i,$ be a lattice, and $\Gamma^*:=\{ x\in \mathbb R^d : \langle x, \gamma\rangle \in 2\pi \mathbb Z \text{ for all } \gamma \in \Gamma\}=\sum_{i=1}^d \ZZ v_i$ be the dual lattice. We denote by $U(\mathcal H)$ the group of unitary operators on $\mathcal H$ and by $L(\mathcal H)$ the space of bounded linear operators on $\mathcal H$. 

The lattices $\Gamma$ and $\Gamma^*$ naturally give rise to a torus as well as a dual torus, that is, $\mathbb R^d/\Gamma$ and $\mathbb R^d/\Gamma^*$. However, for simplicity, we will often just say that we have a vector bundle over a torus $\mathbb T^d$, which then in the physics setting corresponds to the dual torus $\mathbb R^d /\Gamma^*$.

\smallsection{Structure of the paper}
In Section \ref{sec:Chern}, we review the (properties of) characteristic classes for real and complex vector bundles that are relevant to fragile topology. In Section \ref{sec:Wannier}, we introduce Wannier functions and discuss their properties under the compatibility condition \eqref{eq:compatible} with the $I$-symmetry. In Section \ref{sec:Frag_topo}, we discuss the splitting of real vector bundles and its relation with the existence of exponentially decaying Wannier functions. Using these, we prove Theorem \ref{theo:main1} and \ref{theo:main2}. In Section \ref{sec:TBG}, we study the fragile topology phenomenon in the twisted bilayer graphene using the framework we developed. 

\smallsection{Acknowledgment} We would like to thank Junyeong Ahn who very kindly explained to us his work \cite{apy19}, and Maciej Zworski for many helpful discussions and suggestions on the manuscript, especially for pointing out Proposition \ref{def:WB-c} to the authors. We would like to thank Qiuyu Ren for explaining to us the classification of vector bundles over the torus. We also thank Patrick Ledwith, Lin Lin, Solomon Quinn, Kevin Stubbs, Oskar Vafek, and Alexander Watson for discussions on fragile topology. ZT and MY were partially supported by the Simons Targeted Grant Award No.~896630. MY acknowledges the support of AMS-Simons travel grant. SB acknowledges support by SNF Grant PZ00P2 216019.

\section{Characteristic classes and bundle classification}
\label{sec:Chern}

In this section, we discuss the classification of complex vector bundles $\mathcal{E}_{\CC}$ and real vector bundles $\mathcal{E}_{\RR}$ over a manifold of dimension $d\leq 3$. We also recall the basic properties of Chern classes, Euler classes, and Stiefel--Whitney classes. The classification of real vector bundles is well-established in topology. We provide a self-contained discussion for the sake of completeness. For the classification of topological insulators from a different perspective, see \cite{ng}.

\subsection{Bundles}\label{sec:bundle}
We briefly recall the definition of vector bundles, the notion of orientation, and the classification of line bundles.
\begin{defi}[Vector bundle]
\label{def:vb}
    Let $E, X$ be topological spaces. $\pi: E\to X$ is called a (complex or real) \emph{vector bundle} of rank $r$ if for any $x\in X$, $\pi^{-1}(x)$ is a (complex or real) vector space of dimension $r$, and there exists a covering $\{U_i\}$ of $X$ such that there is a homeomorphism, called the trivialization, which is linear on each fiber $\pi^{-1}(x)$, such that the following diagram commutes 
    \[
    \begin{tikzcd}
    \pi^{-1}(U_i) \arrow[r, "\cong"] \arrow[d, "\pi"] & U_i \times \mathbb{C}^r \text{ or } U_i \times \mathbb{R}^r \arrow[ld, "pr_1"] \\
    U_i &
    \end{tikzcd}
    \]
    Here $E$ is also called the \emph{total space} and $X$ is called the \emph{base}. A vector bundle of rank $1$ is called a line bundle.
\end{defi}
\begin{rem}\label{rem:analytic}
     In the above definition, if $E,X$ are $C^s$ manifolds, $\pi$ is $C^s$ and the trivialization is a $C^s$ map, then the vector bundle is called a $C^s$ vector bundle, where $s\in \NN$ or $s=\infty$ (smooth) or $s=\omega$ (real analytic). We note that the classifications of $C^s$ vector bundles are equivalent for $s\in\NN$ or $s=\infty$ or $s=\omega$, see \cite[Theorem 5]{S64}.

\end{rem}
Replacement of the vector space with more general objects leads to the definition of fiber bundles.
\begin{defi}[Fiber bundle]
A \textit{fiber bundle} is a quadruple \( (E, B, \pi, F) \), where \( E \) is the \textit{total space}, \( B \) is the \textit{base space}, \( \pi: E \to B \) is a continuous \textit{projection map}, and \( F \) is the \textit{typical fiber}, a topological space, satisfying the following conditions:
\begin{enumerate}
    \item For every \( b \in B \), the preimage \( \pi^{-1}(\{b\}) \) is homeomorphic to \( F \),
    \item There exists an open cover \( \{U_i\}_{i \in I} \) of \( B \) and homeomorphisms \( \phi_i: \pi^{-1}(U_i) \to U_i \times F \) such that \( \pi(e) = \operatorname{pr}_1(\phi_i(e)) \), where \( \operatorname{pr}_1: U_i \times F \to U_i \) is the projection onto the first component.
\end{enumerate}
\end{defi}
From the definition of a fiber bundle, we can now define the principal $G$-bundles:
\begin{defi}[G-bundle]
A principal \( G \)-bundle, where $G$ is a topological group, is a fiber bundle such that $G$ acts continuously from the right on $E$. In addition, the action of the group must leave the fibers $\pi^{-1}(x)$ invariant and act freely and transitively on them.
\end{defi}
The extension of this concept to the smooth or real analytic category is straightforward by requiring that $G$ has a smooth/real analytic structure and the action is smooth/real analytic.

The reason for introducing the principal $G$-bundles is that they allow the classification of complex or real vector bundles with symmetries. Complex or real vector bundles naturally correspond to principal $G$-bundles, where $G$ is the general linear group (real or complex). One can extend the correspondence to include other groups $G$, by adding more structure to the vector bundle, such as, for instance, orientation:

\begin{defi}[Orientation]
    An \emph{orientation} of a real vector space is an equivalence class of ordered bases, where two ordered bases are equivalent if the invertible matrix taking the first basis to the second has a positive determinant. A vector bundle is called \emph{orientable} if there exists a continuous choice of orientation on each fiber.
\end{defi}

We now briefly recall the classification of real or complex line bundles, which shall be used in later sections.
\begin{enumerate}
    \item For a structure group $G$, $G$-principal bundles are classified by the classifying space $BG$, which is a principal $G$-bundle whose total space $EG$ is contractible:
    \begin{equation*}
        \{\text{principal $G$-bundles over } M\}/\text{isomorphism} \cong [M, BG]
    \end{equation*}
    where $[M,BG]$ is the set of homotopy classes of continuous maps from $M$ to $BG$.
    \item In particular, complex vector bundles of rank $r$ are classified by $BU(r)$. Real vector bundles of rank $r$ are classified by $BO(r)$ and oriented real vector bundles of rank $r$ are classified by $BSO(r)$.
    \item We have $BU(1)=\mathbb{CP}^{\infty}=K(\ZZ,2)$, $BO(1)=\mathbb{RP}^{\infty}=K(\ZZ/2,1)$, where $K(A,n)$ is the Eilenberg--Maclane space which satisfies
    \begin{equation*}
        \pi_n(K(A,n))=A, \quad \pi_m(K(A,n))=1, \quad m\neq n.
    \end{equation*}
    It has the universal property 
    \begin{equation*}
        [M,K(A,n)] \cong H^n(M;A).
    \end{equation*}
    \item 
    The element of $[M,BU(1)]\cong H^2(M;\ZZ)$ is the first Chern class $c_1$. The element of $[M, BO(1)]\cong H^1(M;\ZZ/2)$ is the first Stiefel--Whitney class $w_1$.
\end{enumerate}

\subsection{Classification of complex vector bundles}
For a rank $r$ complex vector bundle $\mathcal{E}_{\CC}$ over a manifold $M$, the Chern class $c_k(\mathcal{E}_{\CC
})$ is an element of $H^{2k}(M;\ZZ)$.
If the vector bundle has a smooth connection with curvature $\Omega$ which is a $\mathfrak{gl}(r,\CC)$-valued $2$-form, the Chern classes can be computed by
\begin{equation*}
    \det\left(I+\frac{\sqrt{-1}}{2\pi}t[\Omega]\right) =\sum_{j\geq 0} c_j(\mathcal{E}_{\CC}) t^j.
\end{equation*}

Over any manifold $M$, complex line bundles are classified up to isomorphisms by the first Chern class $c_1 \in H^2(M;\ZZ)$ as discussed at the end of Subsection \ref{sec:bundle} (see also \cite[p.~34]{C79} for a different proof). For higher-rank complex vector bundles, we have the following classification:
\begin{prop}
\label{prop:chern}
    Complex vector bundles $\mathcal{E}_{\CC}$ over a manifold $M$ of dimension $\leq 3$ are classified by the first Chern class $c_1 \in H^2(M;\ZZ)$. Furthermore, $\mathcal{E}_{\CC}$ admits a decomposition into line bundles 
    \begin{equation}\label{eq:complex-bundle-decomp}
        \mathcal{E}_{\CC} = \bigg(\bigoplus_{i=1}^{r-1} L_i\bigg) \oplus L,
    \end{equation}
    where all line bundles $L_i$ are trivial apart from possibly $L$ for which $c_1(\mathcal{E}_{\CC})=c_1(L).$
\end{prop}

\begin{proof}
Let $\mathcal{E}_{\CC}$ be a complex vector bundle of rank $r \geq 2$ over $M$. By the Thom transversality theorem, there exists a non-vanishing section $s$ of $\mathcal{E}_{\CC}$. Since if $s$ intersects the zero section $0_M$ transversally at $(x,0)\in\mathcal{E}_{\CC}$, we have 
\[ ds (T_xM) + T_{(x,0)}0_M = T_{(x,0)}\mathcal{E}_{\CC}.\]
Thus, the dimension of the intersection of the tangent space
\[
\dim_{\RR}(ds (T_xM) \cap T_{(x,0)}0_M) = \dim_{\RR}(ds (T_xM)) + \dim_{\RR}(T_{(x,0)}0_M) - \dim_{\RR}( T_{(x,0)}\mathcal{E}_{\CC}) < 0,
\]
ensures that the intersection is empty. Hence, $s$ induces a trivial line bundle $L_1$, and $\mathcal{E}_{\CC}$ splits as a Whitney sum:
\[
\mathcal{E}_{\CC} \cong L_1 \oplus \mathcal{E}_{\CC}',
\]
where $\mathcal{E}_{\CC}'$ is a complex vector bundle of rank $r-1$. By induction, $\mathcal{E}_{\CC}$ is isomorphic to the Whitney sum of $r-1$ trivial line bundles with a complex line bundle $L$. The Chern number of $L$ satisfies $c_1(L) = c_1(\mathcal{E}_{\CC})$, which completes the proof.
\end{proof} 
\begin{rem}
    When $\mathcal{E}_{\CC}$ is real analytic, the decomposition in the previous Proposition can be chosen so that $L_i$ and $L$ are real analytic line bundles. This is because \eqref{eq:complex-bundle-decomp} holds in the topological category, and Remark~\ref{rem:analytic} implies a real analytic isomorphism
    \begin{equation*}
         \mathcal{E}_{\CC} \cong \bigg(\bigoplus_{i=1}^{r-1} L_i\bigg) \oplus L.
    \end{equation*}
\end{rem}

\subsection{Stiefel--Whitney class}
For a real vector bundle $\mathcal{E}_{\RR}$ of finite rank on a base space $M$, the $i$-th Stiefel--Whitney class $w_i(\mathcal{E}_{\RR})$ is an element of $H^i(M;\ZZ/2)$. An axiomatic definition of the Stiefel--Whitney class can, for instance, be found in \cite[Theorem~5.4]{D94}.
Some basic properties of the Stiefel--Whitney classes that will be relevant to us are:
\begin{enumerate}
    \item $w_i(\mathcal{E}_{\RR})=0$ if $i >\operatorname{rank}(\mathcal{E}_{\RR})$ or $i>\dim M$.
    \item $\mathcal{E}_{\RR}$ is orientable if and only if $w_1(\mathcal{E}_{\RR})=0$ \cite[Theorem~2.1]{D94}. 
    \item Let
\begin{equation*}
    w_t(\mathcal{E}_{\RR})= 1+w_1(\mathcal{E}_{\RR})t+\cdots+w_r(\mathcal{E}_{\RR})t^r,\quad w_i(\mathcal E) \in H^i(M;\ZZ/2).
\end{equation*} 
Then the Stiefel--Whitney class of the Whitney sum is given by 
\begin{equation*}
    w_t(\mathcal{E}_{\RR}\oplus \mathcal{F}_{\RR}) =w_t(\mathcal{E}_{\RR})\smile w_t(\mathcal{F}_{\RR}),
\end{equation*}
where the cup product induces a bilinear operation on cohomology via the binomial formula:
\[H^i(M;\ZZ/2)\times H^{k-i}(M;\ZZ/2) \to H^k(M;\ZZ/2).\]
Consequently, we have $w_k(\mathcal E_{\RR}\oplus \mathcal F_{\RR})=\sum_{i=0}^k w_i(\mathcal E_{\RR})\smile w_{k-i}(\mathcal F_{\RR})$ and in particular
\begin{gather}
   \label{eq:SW-sum}
    w_1(\mathcal E_{\RR} \oplus \mathcal F_{\RR}) = w_1(\mathcal E_{\RR}) + w_1(\mathcal F_{\RR}), \\
    w_2(\mathcal E_{\RR} \oplus \mathcal F_{\RR}) = w_2(\mathcal E_{\RR}) + w_1(\mathcal E_{\RR}) \smile w_1(\mathcal F_{\RR}) + w_2(\mathcal F_{\RR}). 
\end{gather}
\item When tensoring two real line bundles $L_1,L_2$, the first Stiefel--Whitney class is given by $w_1(L_1\otimes L_2)= w_1(L_1) + w_1(L_2)$. In general, we can assume that both bundles are direct sums of line bundles and state a general formula. This formula still holds in general by the ``splitting principle" (see~\cite[Section~21]{BT}).

For example, if $\mathcal{E}_{\RR}$ is a real vector bundle of rank $r$ and $L$ is a real line bundle, then
\begin{equation}
\label{eq:SW-tensor}
    w_1(\mathcal{E}_{\RR}\otimes L) = w_1(\mathcal{E}_{\RR}) +r w_1(L).
\end{equation}
In particular, the tensor bundle $L\otimes L$ is always a trivial line bundle.
\end{enumerate}

\subsection{Euler class}
Let $\mathcal{E}_{\RR}$ be a real oriented vector bundle of rank $2k$ over a smooth manifold $M$, the Euler class $e(\mathcal{E}_{\RR})$ is an element of $H^{2k}(M;\ZZ)$.

If $\mathcal{E}_{\RR}$ has an orthogonal connection (with respect to a bundle metric) with curvature $\Omega$ (an $\mathfrak{so}(2k)$-valued $2$ form), the Euler class is given by
\[
e(\mathcal{E}_{\RR}) = \frac{1}{(2\pi)^k} [\operatorname{Pf}(\Omega)] \in H^{2k}_{\text{dR}}(M;\mathbb{R}),
\]
where $\operatorname{Pf}(\Omega)$ is the Pfaffian of $\Omega$. \begin{ex}
In particular, for $M = \mathbb R^2 /\Gamma^*$, $k=1$, $d=2$ and a connection $\Omega = \begin{pmatrix} 0 & \Omega_{12} \\ -\Omega_{12} & 0 \end{pmatrix},$ the Euler number is
\[ \chi(\mathcal E_{\mathbb R}) := \int_{\mathbb R^2/\Gamma^*}  e(\mathcal E_{\mathbb R})=\frac{1}{2\pi}\int_{\mathbb R^2/\Gamma^*} \Omega_{12}.\]
\end{ex}
There is a natural homomorphism $\sigma :H^{2k}(M;\ZZ)\to H^{2k}(M;\ZZ/2)$ that maps the Euler class to the top Stiefel--Whitney class \cite[Proposition~9.5]{MS74} 
\begin{equation}
    \label{eq:mod2}
    w_{2k}(\mathcal E_{\RR})= \sigma (e(\mathcal E_{\RR})).
\end{equation}
The Euler class is also related to the Chern class. In fact, by \cite[Lemma 14.1]{MS74} for a complex vector bundle $\mathcal E_{\CC}$ with rank $k$, its underlying real vector bundle $\mathcal E_{\RR}$ of rank $2k$ has a canonical orientation. The top Chern class $c_{k}(\mathcal E_{\CC})\in H^{2k}(M;\ZZ)$ then coincides with the Euler class $e(\mathcal E_{\RR})$ (cf.~\cite[Section 14.2]{MS74}).

\subsection{Classification of real vector bundles}
Over any manifold $M$, real line bundles are classified by the first Stiefel--Whitney class $w_1\in H^1(M;\ZZ/2)$ as discussed at the end of Subsection \ref{sec:bundle}. 

Suppose $M$ has dimension $\leq 3$, we can classify real vector bundles over $M$. 
\begin{prop}
\label{prop:real_class}
    Real vector bundles of rank $r$ over a manifold $M$ with dimension $\leq 3$ are classified by the Stiefel--Whitney classes $w_1\in H^1(M;\ZZ/2)$ and $w_2\in H^2(M;\ZZ/2)$ when $r\neq 2$. When $r=2$, oriented real vector bundles are classified by the Euler class $e_2\in H^2(M;\ZZ)$.
\end{prop}

\begin{proof}
Then rank $2$ oriented real bundles are classified by the Euler class
\begin{equation*}
    e_2\in[M, BSO(2)]=[M, K(\ZZ,2)]=H^2(M;\ZZ).
\end{equation*}
For rank $r$ real vector bundles with $r\geq 3$, we may assume $r=3$ because for $r\geq 4$ we can always split off a trivial subbundle as in the proof of Proposition~\ref{prop:chern}. Moreover, we may assume it is orientable by tensoring with a line bundle, cf.~\eqref{eq:SW-tensor}. Now oriented rank $3$ line bundles are classified by $[M,BSO(3)]$. Since $\pi_1 BSO(3)=\pi_0 SO(3)=1$, $\pi_2 BSO(3)=\pi_1 SO(3)=\ZZ/2$, $\pi_3 BSO(3)=\pi_2 SO(3)=1$, they are classified by the second Stiefel--Whitney class $w_2$:
\begin{equation*}
    [M, BSO(3)]=[M, K(\ZZ/2,2)]=H^2(M;\ZZ/2)\ni w_2.
\end{equation*}
Since the Stiefel--Whitney class does not depend on the orientation, this finishes the classification of rank $3$ vector bundles.
\end{proof}

\section{Wannier functions with symmetry constraints}
\label{sec:Wannier}
In this section, we recall the basics of Wannier functions and discuss the compatibility of Wannier functions with the $I$-symmetry.  

\subsection{Preliminaries on Wannier functions}
\label{sec:prelim}
We first recall the definition of Wannier functions. For the projections $P(k)$ with rank $r$ satisfying Assumption \ref{ass:proj} and the Bloch bundle $\mathcal{E}_{\CC}$, we can define a basis of $\operatorname{ran} P(k)$ for ${k\in \mathbb R^d/\Gamma^*}$.
\begin{defi}
\label{def:bb}
For the orthogonal projections $\mathcal P:=(P(k))_{k \in \mathbb R^d}$ satisfying Assumption \ref{ass:proj}, we say that $\Phi:\mathbb R^d \to \mathcal H^r$ is a \emph{global Bloch frame}, if 
   \[\Phi : \RR^d \to \mathcal{H} \oplus \dots \oplus \mathcal{H} = \mathcal{H}^r,\quad k \mapsto (\varphi_1(k), \dots, \varphi_r(k)),\]
   is \emph{$\tau$-equivariant}:
    \begin{equation}
        \label{eq:equivariance}
        \varphi_a(k + \gamma) = \tau(\gamma) \varphi_a(k) \quad \text{for all } k \in \mathbb{R}^d, \, \gamma \in \Gamma^*, \, a \in \{1, \dots, r\},
    \end{equation}
and for a.e.\ $k \in \RR^d$, the set $\{\varphi_1(k), \dots, \varphi_r(k)\}$ is an orthonormal basis spanning $\operatorname{Ran} P(k)$. 
\end{defi}

\begin{rem}
\label{rem:section}
    In the language of vector bundles, each $\varphi_a(k)$ is a normalized section of the Bloch bundle, and a Bloch frame $\Phi$ is a family of orthonormal sections of the Bloch bundle $\mathcal{E}_\CC$ that span the Bloch bundle over $\CC$ for a.e. $k\in \RR^d$. 
\end{rem}

We now define Wannier functions. First, recall the \textbf{Bloch transform} is given by 
 \begin{equation}
     \label{eq:BFT}
     (\mathcal B\psi)(k,x) :=\sum_{\gamma \in \Gamma} e^{-i\langle x+\gamma,k\rangle} \psi(x+\gamma), \quad \psi\in L^2(\mathbb R^d), 
 \end{equation}
which is an isometry 
 \[\mathcal B: L^2(\mathbb R^d) \longrightarrow \mathcal{H}_\tau := \{\varphi \in L^2_{\text{loc}}(\RR^d;L^2(\mathbb R^d/\Gamma )): \varphi(k+\gamma,x)=\tau(\gamma)\varphi(k,x) \text{ for all }\gamma \in \Gamma^*\}. \]
\begin{defi}
\label{def:Wannier}
Let $\varphi(k,x) \in \mathcal H_{\tau}$ be a normalized section of the Bloch bundle, the \emph{Wannier function} $w(\varphi) \in L^2(\mathbb R^d)$ is defined by
\begin{equation}
\label{eq:Wannier}
 w(\varphi)(x):=\frac{1}{\vert \mathbb R^d / \Gamma^* \vert} \int_{\mathbb R^d / \Gamma^*} e^{i\langle x, k\rangle}\varphi(k,x) \ dk=\mathcal B^{-1}(\varphi). 
\end{equation}
Moreover, Wannier functions of the Bloch frame $\Phi = (\varphi_a)_{1\leq a\leq r}$ are given by $\{w(\varphi_a)\}_{1\leq a\leq r}$.
\end{defi}

For a global Bloch frame $\Phi = (\varphi_a)_{1\leq a\leq r}$ of an orthogonal projection $\mathcal{P}$, the shifted Wannier functions
\begin{equation}
\label{eq:wb2}
\mathcal{W} = \{w(\varphi_a)(\bullet - \gamma)\}_{\gamma \in \Gamma, 1\leq a\leq r}   
\end{equation}
form an orthonormal basis of the space 
$$\bigoplus_{k\in\RR^d/\Gamma^*}\operatorname{ran} P(k).$$ 
Hence, the set $\mathcal{W}$ is also a \emph{Wannier basis} as in the Definition \ref{def:wb}. See \cite[Section~6.5]{ku2} or \cite[Proposition~5.5]{notes} for detailed discussions. 

\subsection{Chern numbers and localization dichotomy}
\label{sec:Chern_review}
We recall the criterion of exponential localization of Wannier functions without symmetry constraints. 

Recall that $P(k)$ denotes the family of orthogonal projections parametrized by $k \in \mathbb{R}^d$ that satisfy the equivariance conditions \eqref{eq:t-equivariance} of the reciprocal lattice $\Gamma^*$. This gives rise to a Bloch bundle $\mathcal{E}_\CC$ over the torus $\RR^d/\Gamma^*$ defined in \eqref{eq:bb}, equipped with a Berry connection.   
Let $\Omega(k) = \sum_{i,j} \Omega_{ij}(k) \ dk_i \wedge dk_j$ be the curvature form of the Berry connection:
\[ \Omega_{ij}(k):=\tr_{\mathcal H}(P(k)[\partial_i P(k),\partial_j P(k)]).\]
Then, the first Chern class is defined as
\[ c_1(\mathcal E_{\CC})=\tfrac{i}{2\pi} [\Omega] \in H^2_{\text{dR}}(\mathbb R^d/\Gamma^*;\RR),\]
where $[\Omega]$ represents the de Rham cohomology class of the curvature form $\Omega$.

To compute the first Chern class explicitly, we use the de Rham isomorphism. For the $d$-dimensional torus $\mathbb{R}^d / \Gamma^*$, the second homology group satisfies
\[
H_2(\mathbb R^d/\Gamma^*) \cong \mathbb Z^{\binom{d}{2}}.
\]
This group is generated by the independent $2$-cycles $(B_{ij})_{1\le i < j \le d},$ where $B_{ij}:=(\mathbb R v_i + \mathbb R v_j)/(\ZZ v_i + \ZZ v_j)$ and $(v_i)_{i=1}^d$ are basis vectors of $\Gamma^*$.

\begin{defi}
    A family of projections $\mathcal P$ that satisfy Assumption~\ref{ass:proj} is called \emph{Chern trivial} if for $d \in \{2,3\}$ the Chern numbers
    \[ C_1(\mathcal E_{\CC})_{ij}:=\frac{i}{2\pi} \int_{B_{ij}} \tr_{\mathcal H}(P(k)[\partial_i P(k),\partial_j P(k)]) \ dk_i \wedge dk_j\]
    vanish for all $1\le i<j\le d$.
\end{defi}

By the de Rham isomorphism, Chern triviality implies the vanishing of the first Chern class $c_1(\mathcal E_{\CC})=0.$ The first Chern class has direct implications on the structure of Wannier functions. This is the \emph{localization dichotomy} as stated, for instance, in \cite{MPPT18}:
\begin{itemize}
    \item \textbf{Case 1:} If $\mathcal P$ is Chern-trivial, then there exists a Bloch frame $\Phi$ that is real analytic in the parameter $k$, such that all Wannier functions $w(\varphi_a)$ decay exponentially.
    \item \textbf{Case 2:} If $\mathcal P$ is not Chern-trivial, then there does not exist a Bloch frame that is $H^1_{\loc}$ in $k$. In addition, for any Bloch frame $\Phi$, there exist Wannier functions $w(\varphi_a)$ for some $a$ such that $ x w(\varphi_a) \notin {L^2} $.
\end{itemize}

It is instructive to compare this dichotomy with Theorem~\ref{theo:main1} and \ref{theo:main2}. If $\mathcal E_{\mathbb C}$ is not Chern-trivial and $\mathcal E_{\mathbb C}'$ is Chern-trivial, then the direct sum $\mathcal E_{\mathbb C} \oplus \mathcal E_{\mathbb C}'$ remains not Chern-trivial. Hence, the obstruction to the existence of exponentially localized Wannier basis imposed by the Chern number cannot be removed by adding trivial bundles. This highlights the distinction between \emph{stable topology} in condensed matter physics, where obstructions cannot be lifted by modifications by topologically trivial elements, and \emph{fragile topology}, which depends on the rank and can be lifted by changing it.

\subsection{Compatibility of Wannier functions}
\label{sec:compatibility}
Now we consider the Bloch bundle $\mathcal{E}_{\CC}$ and Wannier basis $\mathcal{W}$ equipped with the symmetry $Iu(z)= \overline{u(-z)}$ defined in \eqref{eq:I}. 

By Assumption \ref{ass:I}, the $I$-symmetry acts as a real vector bundle homomorphism (i.e., each fiber is preserved). Since $I^2=1$, it has eigenvalue $1$ and $-1$ on each fiber. The $I$-symmetry naturally induces a rank $r$ real subbundle
$$ \mathcal{E}_{\RR}  =\{(k,v)\in \mathcal{E}_{\mathbb{C}}: I v=v\}\subset \mathcal{E}_{\CC},$$ 
given by the eigenspace of $1$. 

We now give an equivalent characterization of Wannier functions $\mathcal{W}$ compatible with the $I$-symmetry, i.e., $I\mathcal{W} = \mathcal{W}$. We shall see that this characterization relates the compatibility conditions to the Stiefel--Whitney classes of the real line bundles $\mathcal{E}_{\RR}$. 

\begin{prop}
\label{def:WB-c}
A Wannier basis $\mathcal{W}$ is compatible with the $I$-symmetry in the sense of Definition \ref{def:compatible} if and only if there exist functions $\{\psi_{a}(x)\in L^2(\RR^d): a\in\{1,\cdots, r\}\}$ and $\mathfrak{c}_1,\cdots, \mathfrak{c}_r\in \Gamma$ such that
\begin{equation}
\label{eq:WB-c}
    I\psi_{a,\gamma} = \psi_{a,\mathfrak{c}_{a}-\gamma},\quad a\in\{1, \cdots , r\}, \quad \psi_{a,\gamma}(x):=\psi_a(x-\gamma)
\end{equation}
and
\begin{equation}
    \operatorname{span}_{\CC} \mathcal{W} = \operatorname{span}_{\CC} \{\psi_{a,\gamma}(x): \gamma\in\Gamma, a\in\{1,\cdots, r\}\}.
\end{equation}
In particular, the points $\mathfrak{c}_1,\cdots, \mathfrak{c}_r\in \Gamma$ are called the \emph{rescaled Wannier centers}.
\end{prop}
\begin{proof}
The ``if" direction is obvious. We prove the ``only if" direction. Suppose $\mathcal{W}$ is a Wannier basis compatible with the $I$-symmetry and $\psi_{a,0}\in\mathcal{W}$. By assumption, there exists $\psi_{b,\mathfrak{c}}\in \mathcal{W}$ such that
\begin{equation*}
    I\psi_{a,0}=\psi_{b,\mathfrak{c}}.
\end{equation*}
If $a=b$, then by the definition of $I$, we have
\begin{equation*}
    I\psi_{a,\gamma}(x)=I(\psi_{a,0}(\cdot-\gamma))=(I\psi_{a,0})(x+\gamma)=\psi_{a,\mathfrak{c}-\gamma}(x).
\end{equation*}
If $a\neq b$, then we let
\begin{equation*}
    \tilde{\psi}_{a,0}:=\frac{1}{\sqrt{2}}\left(\psi_{a,0}+\psi_{b,\mathfrak{c}}\right),\quad \tilde{\psi}_{b,0}:=\frac{i}{\sqrt{2}}\left(\psi_{a,0}-\psi_{b,\mathfrak{c}}\right).
\end{equation*}
Then $I\tilde{\psi}_{a,0}=\tilde{\psi}_{a,0}$, $I\tilde{\psi}_{b,0}=\tilde{\psi}_{b,0}$ and $\{\tilde{\psi}_{a,\gamma},\tilde{\psi}_{b,\gamma}\}_{\gamma\in\Gamma}$ give an orthonormal basis of the space $\overline{\mathrm{span}_{\CC}(\psi_{a,\gamma},\psi_{b,\gamma})}$. We can apply this process to the remaining basis and get \eqref{eq:WB-c}.
\end{proof}

By taking the Bloch transform \eqref{eq:BFT} with $\varphi_a = \mathcal{B}\psi_a$, equation \eqref{eq:WB-c} translates into
\begin{equation}
\label{eq:conditions}
    I\varphi_a(k,x)=e^{i\langle\mathfrak{c}_a,k\rangle}\varphi_a(k,x), \quad a\in\{1, \cdots r\}.
\end{equation} 
Hence, each $\varphi_a(k,\cdot)$ is a section of the real subbundle
\begin{equation}
\label{eq:realB-c}
    \mathcal{E}_{\mathfrak{c}_a} :=\{(k,v)\in \mathcal{E}_{\mathbb{C}}: I v= e^{i\langle \mathfrak{c}_a,k\rangle}v\}\subset \mathcal{E}_{\CC},
\end{equation}
where $\mathcal{E}_{\mathfrak{c}}$ is a subbundle of the complex vector bundle $\mathcal{E}_{\mathbb{C}} \cong \mathcal{E}_{\RR}\otimes \CC$. 

\subsection{Wannier centers}
We continue with a more detailed discussion on Wannier centers. 
For the real subbundle $\mathcal{E}_{\mathfrak{c}}$ defined in \eqref{eq:realB-c}, 
We claim
\begin{equation}
\label{eq:tensorB-c}
    \mathcal{E}_{\mathfrak{c}}\cong \mathcal{E}_{\RR}\otimes \mathcal{L}_{\mathfrak{c}}
\end{equation}
where $\mathcal{L}_{\mathfrak{c}}$ is the real line bundle, corresponding to the Stiefel--Whitney class $\mathfrak{c}\in \Gamma/(2\Gamma)\cong H^1(\RR^n/\Gamma^*;\ZZ/2)$. In other words, $\mathcal{L}_{\mathfrak{c}}$ is a real line bundle contained in the trivial complex line bundle $\TT^d\times \CC$ satisfying
\begin{equation}\label{eq:line-bundle-Lc}
    \mathcal{L}_{\mathfrak{c}} := \{(k,z)\in \TT^d\times \CC: \overline{z}=e^{i\langle \mathfrak{c},k\rangle}z\}.
\end{equation}
The isomorphism \eqref{eq:tensorB-c} follows from identifying $\mathcal{E}_{\CC}$ with $\mathcal{E}_{\RR}\otimes \CC$, where $I$ acts trivially on $\mathcal{E}_{\RR}$ and acts as a conjugation on $\CC$. In particular, $\mathcal{E}_{\mathfrak{c}} \cong \mathcal{E}_{\mathfrak{c'}}$ if $\mathfrak{c}-\mathfrak{c}'\in 2\Gamma$ as $\mathcal{L}_{\mathfrak{c}}$'s are classified by the first Stiefel--Whitney class $w_1\in H^1(\RR^n/\Gamma^*;\ZZ/2)\cong \Gamma/(2\Gamma)$. 

We now give a more physically relevant definition of the Wannier center. 
\begin{defi}
\label{def:wc}
    Let $\langle \bullet \rangle^{1/2} w(\varphi) \in L^2(\mathbb R^d)$. The \emph{Wannier center} $\mathfrak{c}(w(\varphi))$ is defined as the expectation of the position operator
    \[ \mathfrak{c}(w(\varphi)):=\int_{\mathbb R^d} \vert w(\varphi)\vert^2 x \,dx \in \mathbb R^d.\]
\end{defi}
Since Wannier centers are the expectation of the position of charge carriers, they are directly related to the polarization where $q$ is the electric charge of the particles \cite{MV,M+12}
\[ P = q\mathfrak{c}(w(\varphi)).\]
We now clarify the relation between the two definitions of (rescaled) \emph{Wannier centers} as the expectation of the position in Definitions \ref{def:wc} and phase factors (of the $I$-symmetry) in Proposition  \ref{def:WB-c}. In particular, Wannier centers give rise to the first Stiefel--Whitney class of the associated real line bundle \eqref{eq:realB-c}. 

\begin{prop}
\label{prop:Berry}
    Let $u \in \mathcal{H}_{\tau} \cap H^{\frac{1}{2}}_{\text{loc}}(\mathbb R^d;L^2(\mathbb R^d/\Gamma))$ be a normalized section. Suppose $Iu_k = e^{i\langle \mathfrak c,k\rangle}u_k$ (as in Proposition \ref{def:WB-c}) is anti-unitary. Then for $\mathfrak c(w(u))$ as in Definition \ref{def:wc},
    \[ \mathfrak c(w(u))=\frac{\mathfrak c}{2}.\]
\end{prop}
\begin{proof}
For simplicity, we shall give our proof in the dimension $d=2$. We start by recalling that after conjugating by the Bloch--Floquet transform $\mathcal U x \mathcal U^{-1} = i\nabla_k$ which shows that 
\[ \begin{split} \mathfrak{c}(w(u)) &= \langle w(u),x w(u) \rangle_{L^2(\RR^d)} = \frac{\langle \mathcal U w(u),(\mathcal U x \mathcal U^{-1}) \mathcal U w(u) \rangle_{L^2(\mathbb R^d/\Gamma \times \mathbb R^d/\Gamma^*)}}{ \vert v_1 \wedge v_2 \vert } \\
&=  -\frac{i \langle u,\nabla_k u \rangle_{L^2(\mathbb R^d/\Gamma \times \mathbb R^d/\Gamma^*)}}{ \vert v_1 \wedge v_2 \vert }, \end{split}\]
where  $\Gamma^* = \mathbb Z v_1 + \mathbb Zv_2$. With the Berry connection $A(k)=-i\langle u_k,\nabla_k u_k\rangle_{L^2(\mathbb R^d/\Gamma )}$ and unit vector $e_{v} = \frac{v}{\Vert v \Vert}$
\[\begin{split} \langle \mathfrak{c}(w(u)), e_{v_1} \rangle &= \frac{\vert v_1 \wedge v_2 \vert}{ \vert v_1 \wedge v_2 \vert} \int_0^1 \int_{0}^1 \langle A(t_1v_1 + t_2 v_2),e_{v_1} \rangle \ dt_1 \ dt_{2} \\ 
&=  \int_0^1 \int_{0}^1 \langle A(t_1v_1 + t_2 v_2),e_{v_1} \rangle \ dt_1 \ dt_{2} \\ 
&= \frac{1}{\Vert v_1\Vert} \int_0^1 \int_{0}^1 \langle A(t_1v_1 + t_2 v_2),v_1 \rangle \ dt_1 \ dt_{2} = \frac{1}{\Vert v_1\Vert} \int_0^1 \gamma_{v_1}(t_2) \ dt_{2},
\end{split}  \]
where $\gamma_{v_1}$ is the Berry phase in $v_1$ direction \[ \gamma_{v_1}(t_2) = \int_0^1 \langle A(t_1 v_1+ t_2v_2),v_1 \rangle dt_1.\] The same computation applies to $e_{v_2}$ with $1$ replaced by $2$ everywhere. 

More can be said about the Berry phase when symmetries are enforced. For $Iu_k = e^{i\phi(k)}u_k$ and $\phi(k):=\langle \mathfrak c_{a},k\rangle,$ a simple computation shows that 
\[\begin{split}  \langle u_k,i\nabla_k u_k \rangle  &=-i \langle u_k,\nabla_k u_k \rangle = \left\langle \frac{1}{i} \nabla_k(Iu_k), Iu_k \right\rangle \\
&= \left\langle \frac{1}{i} \nabla_k (e^{i\phi(k)}u_k),  e^{i\phi(k)}u_k \right\rangle = \nabla \phi(k)  -i\langle  \nabla_k u_k,  u_k \rangle.
\end{split}\]
Using that $0=\nabla_k \Vert u_k \Vert^2 = \langle  \nabla_k u_k, u_k\rangle +  \langle  u_k, \nabla_k u_k\rangle$, we find 
\begin{equation}
\label{eq:Berry-SW}
 2\gamma_{v_1} = \int_0^1 \langle \nabla \phi(t_1 v_1+t_2v_2),v_1\rangle \ dt_1 =\langle \mathfrak c_{a},v_1 \rangle 
 \end{equation}
such that 
\[\langle \mathfrak{c}(w(u)), e_{v_j} \rangle = \frac{ \langle \mathfrak c_{a},v_j \rangle}{2 \Vert v_j\Vert} =  \frac{\langle \mathfrak c_{a},e_{v_j} \rangle}{2},\ j=1,2 \Longrightarrow \mathfrak{c}(w(u)) =\frac{\mathfrak{c}_a}{2}. \qedhere \]
\end{proof}

\begin{rem}
It follows from the proof that the Wannier center ${\mathfrak{c}}_a$ in \eqref{eq:realB-c} is completely determined by the Berry phase by \eqref{eq:Berry-SW}.
\end{rem}

\section{Fragile topology}
\label{sec:Frag_topo}
In this section, we consider a complex Bloch bundle $\mathcal{E}_{\CC}$ of rank $r\geq 2$ over $\TT^2$ or $\TT^3$ equipped with the $I$-symmetry in \eqref{eq:I}. Thus we get a real vector subbundle $\mathcal{E}_{\RR}$ with (real) rank $r\geq 2$, i.e.
\begin{equation}\label{eq:bloch-defi}
    \mathcal{E}_{\RR}=\{(x,v)\in \mathcal{E}_{\CC}: Iv=v\},\quad \mathcal{E}_{\CC}=\mathcal{E}_{\RR}\otimes \CC.
\end{equation}
We show that for $r=2$ there is an exponentially localized Wannier basis compatible with the $I$-symmetry if and only if the Euler class $e(\mathcal{E}_{\RR})=0$, whereas for $r\geq 3$ we show that there is always an exponentially localized Wannier basis compatible with the $I$-symmetry.

\subsection{Wannier basis and line bundles}
We first prove an equivalent statement of the existence of an exponentially localized Wannier basis compatible with the $I$-symmetry.
\begin{prop}
\label{prop:Isym}
    The real subbundle $\mathcal{E}_{\RR}$ in \eqref{eq:bloch-defi} over $\mathbb T^d$ can be split into a direct sum of analytic line bundles if and only if there is an exponentially localized Wannier basis $\{\varphi_{\gamma}^{a}(x): \gamma\in\Gamma, a\in\{1,\cdots, r\}\}$ (with some Wannier centers) compatible with the $I$-symmetry. 
\end{prop}

\begin{proof}
    We use the equivalent definition of compatibility in Proposition \ref{def:WB-c}. 
    If $\mathcal{E}_{\RR} = L_1\oplus \cdots \oplus L_r$ is a direct sum of analytic line bundles, we may take the line bundles to be orthogonal. Let $\mathfrak{c}_a\in \Gamma$ be a representative of the Stiefel--Whitney class $w_1(L_a)\in H^1(\RR^d/\Gamma; \ZZ/2) = \Gamma/2\Gamma$. As $L_a\otimes \mathcal{L}_{\mathfrak{c}_a}$ is a trivial line bundle in $\mathcal{E}_{\RR}\otimes \mathcal{L}_{\mathfrak{c}_a}$ and by orthogonality of $\{L_a\}_{1\leq a\leq r}$, there exist analytic orthonormal sections $s_a$ of the tensor bundle $\mathcal{E}_{\RR}\otimes \mathcal{L}_{\mathfrak{c}_a} =\mathcal{E}_{\mathfrak{c}_a}$, for $a=1,2,\cdots, r$. Taking the Bloch transform of $s_a$, we can construct exponentially localized Wannier functions satisfying
    \begin{equation*}
        I \varphi_{\gamma}^{a} = \varphi_{\mathfrak{c}_a-\gamma}^{a},\quad a\in \{1,\cdots, r\}.
    \end{equation*}

    On the other hand, if there exists an exponentially localized Wannier basis with rescaled Wannier centers at $\{\mathfrak{c}_a\}_{1\leq a\leq r}$, then the inverse Bloch transform would give analytic orthonormal sections $s_a$ of bundles $\mathcal{E}_{\mathfrak{c}_a} =\mathcal{E}_\RR\otimes \mathcal{L}_{\mathfrak{c}_a}$, $a=1,2,\cdots, r$. Now we claim $\mathcal{E}_\RR \cong \mathcal{L}_{\mathfrak{c}_1}\oplus \cdots \oplus \mathcal{L}_{\mathfrak{c}_r}$. 

   We can split off a trivial bundle $L_1 \cong \RR s_1$: $$\mathcal{E}_{\mathfrak{c}_1} = L_1 \oplus \mathcal{E}'_{\mathfrak{c}_1}.$$
    Tensoring $\mathcal{E}_{\mathfrak{c}_1}$ with $\mathcal{L}_{\mathfrak{c}_1}$ yields that
    \[\mathcal{E}_{\RR} = (L_1\otimes \mathcal{L}_{\mathfrak{c}_1}) \oplus (\mathcal{E}'_{\mathfrak{c}_1}\otimes \mathcal{L}_{\mathfrak{c}_1}) \cong 
    \mathcal{L}_{\mathfrak{c}_1} \oplus (\mathcal{E}'_{\mathfrak{c}_1}\otimes \mathcal{L}_{\mathfrak{c}_1}).\]
    Repeating this construction gives the desired direct sum decomposition. 
\end{proof}


\subsection{Wannier obstructions for rank two vector bundles}
In this section, we show that an oriented rank two real Bloch bundle $\mathcal{E}_\RR$ over the torus in dimension $d\le 3$ exhibits an exponentially localized Wannier basis compatible with the $I$-symmetry if and only if the Euler class $e(\mathcal{E}_{\RR}) = 0$. In fact, we have the following more general result. 
\begin{prop}
\label{prop:rank2}
    A rank two oriented real vector bundle $\mathcal{E}_\RR$ over $\TT^d$ can be split into the direct sum of line bundles if and only if $e(\mathcal{E}_{\RR}) = 0$. In particular, $\mathcal{E}_\RR$ is a trivial bundle when it splits.
\end{prop}
\begin{proof}
    Suppose $\mathcal{E}_\RR = L_1\oplus L_2$, where $L_1$ and $L_2$ are orthogonal line bundles to each other. Then there exists a section $s_1$ of $L_1$ that vanishes on a subtorus $\gamma_1$ representing $w_1(L_1)$. Similarly, there exists another section $s_2$ of $L_2$ that vanishes on a subtorus $\gamma_2$ representing $w_1(L_2)$. Since $L_1\oplus L_2$ is orientable, $w_1(L_1)=w_1(L_2)$. Therefore, we can choose $\gamma_1$ and $\gamma_2$ to be parallel subtori so that $s_1+s_2$ is nonvanishing everywhere. Thus, $\mathcal{E}_\RR$ contains a non-vanishing section and is trivial.
\end{proof}

\begin{rem}
    This proposition is not true over $\mathbb{RP}^2$. Let $L$ be the tautological line bundle over $\mathbb{RP}^2$, then $L\oplus L$ is not trivial.
\end{rem}

\begin{proof}[Proof of Theorem \ref{theo:main1}]
     By Proposition \ref{prop:Isym}, the existence of an exponentially localized Wannier basis compatible with the $I$-symmetry is equivalent to the decomposability of the real Bloch bundle $\mathcal{E}_\RR$ into a direct sum of line bundles. Proposition \ref{prop:rank2} shows that real bundles of rank two are decomposable to line bundles if and only if the Euler class $e(\mathcal{E}_{\RR})=0$. This shows that \eqref{part1} is equivalent to \eqref{part2}. Obviously \eqref{part2} implies \eqref{part3}.
     
     In order to prove \eqref{part3} implies \eqref{part1}, we note that a Wannier basis compatible with the $I$-symmetry satisfying \eqref{eq:main1-2} gives a normalized section of the real vector bundle $\mathcal{E}_{\RR}\otimes \mathcal{L}$ for some real line bundle $\mathcal{L}$ in $H^1(\TT^d; \mathcal{E}_{\RR}\otimes \mathcal{L})$. Since $\mathcal{E}_{\RR}$ is oriented and has rank $2$, by \eqref{eq:SW-tensor}, $\mathcal{E}_{\RR}\otimes \mathcal{L}$ is also oriented and has rank $2$. We identify it with a complex line bundle, where the Euler class is identified with the Chern class. By \cite[Theorem 7.1]{MPPT18}, this implies the Euler class is trivial, i.e., $\mathcal{E}_{\RR}\otimes \mathcal{L}$ is trivial. This further implies that $\mathcal{E}_{\RR}\cong \mathcal{L}\oplus \mathcal{L}$, which implies $\mathcal{E}_{\RR}$ is trivial by Proposition~\ref{prop:rank2}.
\end{proof}

\subsection{Fragile topology for higher rank vector bundles}
In this section, we consider the so-called fragile topology for three or more Bloch bands over $\TT^2$ or $\TT^3$. 
\begin{prop}
\label{prop:rankg3}
    Let $M$ be $\TT^2$ or $\TT^3$ and $\mathcal{E}_\RR$ be a real vector bundle of rank $\geq 3$ over $M$. Then $\mathcal{E}_\RR$ can be written as a direct sum of real line bundles.
\end{prop}
\begin{proof}
It suffices to consider the case $r=3$, as any real vector bundle with rank $\geq 4$ over $M$ can split off a trivial line bundle by applying the Thom transversality theorem as in Proposition~\ref{prop:chern}.
Without loss of generality, we may also assume $\mathcal{E}_\RR$ is orientable. Otherwise, we may tensor it with a line bundle to make it orientable (see \eqref{eq:SW-tensor}), then we can tensor it with the same line bundle after splitting to obtain the decomposition of the original bundle.   
    
First, we discuss bundles over $\TT^2$. The orientability of $\mathcal{E}_\RR$ yields $w_1=0\in H^1(\TT^2;\ZZ/2)=\ZZ/2 e_1+\ZZ/2 e_2$. We assume $w_2=ae_1\smile e_2\in H^2(\TT^2;\ZZ/2)=\ZZ/2 e_1\smile e_2$. We can choose three line bundles with $w_1=ae_1, ae_2, ae_1+ae_2$ respectively. The bundle we get from the direct sum has the total Stiefel--Whitney class
\begin{equation*}
    (1+ae_1)(1+ae_2)(1+ae_1+ae_2)=1+ae_1\smile e_2,
\end{equation*}
which agrees with the total Stiefel--Whitney class of $\mathcal{E}_\RR$.

Now we discuss bundles over $\TT^3$. Again, the orientability of $\mathcal{E}_\RR$ yields $w_1=0\in H^1(\TT^3;\ZZ/2)=\ZZ/2 e_1+\ZZ/2 e_2+\ZZ/2 e_3$. We assume 
\begin{equation*}
    w_2=a e_1\smile e_2+b e_2\smile e_3+c e_3\smile e_1\in H^2(\TT^3;\ZZ/2)=\ZZ/2 e_1\smile e_2+\ZZ/2 e_2\smile e_3+\ZZ/2 e_3\smile e_1.
\end{equation*}
There are four cases we need to consider
\begin{itemize}
    \item When $a=b=c=0$, this is a trivial bundle.
    \item When one of $a,b,c$ is nonzero, say $a=1, b=c=0$, we choose three line bundles with $w_1=e_1, e_2, e_1+e_2$ respectively, and use
    \begin{equation*}
        (1+e_1)(1+e_2)(1+e_1+e_2)=1+e_1\smile e_2.
    \end{equation*}
    \item When two of $a,b,c$ are nonzero, say $a=b=1, c=0$, we choose three line bundles with $w_1=e_1+e_3, e_2, e_1+e_2+ e_3$ respectively, and use
    \begin{equation*}
        (1+e_1+e_3)(1+e_2)(1+e_1+e_2+e_3)=1+e_1\smile e_2+e_2\smile e_3.
    \end{equation*}
    \item When all of $a,b,c$ are nonzero, i.e.~$a=b=c=1$, we choose three line bundles with $w_1=e_1+e_3, e_2+e_3, e_1+e_2$ respectively, and use
    \begin{equation*}
        (1+e_1+e_3)(1+e_2+e_3)(1+e_1+e_2)=1+e_1\smile e_2+e_2\smile e_3+e_3\smile e_1.
    \end{equation*}
\end{itemize}
In each case, $\mathcal{E}_\RR$ can be decomposed into the direct sum of three line bundles given by the total Stiefel--Whitney class in the LHS of the equalities.
\end{proof}
One may furthermore take the three line bundles to be orthogonal. For real Bloch bundles, this gives the Wannier centers and Wannier functions compatible with the $I$-symmetry.

\begin{proof}[Proof of Theorem \ref{theo:main2}]
    Again by Proposition \ref{prop:Isym}, the existence of an exponentially localized Wannier basis compatible with the $I$-symmetry is equivalent to the decomposability of the real Bloch bundle $\mathcal{E}_\RR$ into a direct sum of line bundles. For $d\leq 3$, Proposition \ref{prop:rankg3} shows that this is always the case for real bundles with rank $\ge 3$.
\end{proof}

\section{Applications: Fragile topology in the chiral TBG at magic angles}
\label{sec:TBG}
In this section, using the fragile topology framework developed in previous sections, we determine the Wannier centers of the exponentially localized Wannier functions arising from the topologically non-trivial flat bands of twisted bilayer graphene (TBG). In particular, we consider the so-called chiral limit of the TBG. See \cite{magic,beta,Z23} for more detailed discussions on the chiral TBG and \cite{BM11,bz23,bmcone} for the general Bistritzer--Macdonald Hamiltonian. 

The Bloch transformed Hamiltonian is given by \[H_k(\alpha) = \begin{pmatrix}
    0 & D(\alpha)^* + \bar k \\ D(\alpha) +k& 0 \end{pmatrix}\text{ with }D(\alpha) = \begin{pmatrix}
        2D_{\bar z} & \alpha U(z) \\ \alpha U(-z) & 2D_{\bar z}
    \end{pmatrix}\] with a potential $U \in C^{\infty}(\CC)$ satisfying for $\gamma \in \Lambda:=\ZZ + \omega \ZZ$ with $\omega  = e^{2\pi i /3}$, $K = \frac{4\pi}{3},$ and $\langle a,b\rangle = \Re(\overline{a}b)$
    \[U(z+\gamma) = e^{i\langle \gamma,K\rangle} U(z), \ U(\omega z) = \omega U(z) \text{ and }\overline{U(\bar z)}=-U(-z), \]such that $H_k$ commutes with the symmetry $I = \begin{pmatrix} 0 & Q \\ Q  &0 \end{pmatrix}$ with $Qu(z)=\overline{u(-z)}.$
    The operator acts on the Hilbert space 
    \[ L^2_0:=\{ v\in L^2_{\text{loc}}(\CC;\CC^4); \mathscr L_{\gamma} u =u \text{ for all }\gamma \in \Lambda\}, \]
    where $L_{\gamma} u(z)=\operatorname{diag}(\omega^{\gamma_1+\gamma_2},1,\omega^{\gamma_1+\gamma_2},1)u(z+\gamma),$ with $\gamma=\gamma_1 + \omega \gamma_2$ and $(\gamma_1,\gamma_2)\in \mathbb Z^2.$

    A \emph{magic angle} in this model is defined as a parameter $\alpha \in \mathbb C$ such that 
    \[ 0\in \bigcap_{k\in \CC} \Spec_{L^2_0}(H_k(\alpha)).\] 
By Fredholm theory, this implies that $\dim \ker_{L^2_0}(D(\alpha)^*+\bar k)=\dim\ker_{L^2_0} (D(\alpha)+k) \neq 0$ for all $k\in \CC.$ If $\dim \ker_{L^2_0}(D(\alpha)+k)=1$ for all $k\in \CC,$ we call the magic angle \emph{simple} and if $\dim \ker_{L^2_0}(D(\alpha)+k)=2$ for all $k\in \CC,$ we call it \emph{two-fold degenerate}. The existence of simple and two-fold degenerated magic angles are proved in \cite{bhz1,bhz3,walu}. 

Recall that flat bands and the corresponding eigenspaces are given by a family of orthogonal projections that satisfy Assumption \ref{ass:proj} due to the existence of band gaps between flat bands and other bands \cite{bhz1,bhz3}. We consider $V(k):=\ker_{L^2_0}(D(\alpha)+k) \subset L^2_0.$
This allows us to define a trivial bundle $\pi:\tilde E \to \CC,$ where
\[ \Tilde E:=\{(k,v):v\in V(k)\} \subset \CC \times L^2_0(\CC/\Lambda;\CC^2).\]
To define a vector bundle over the torus $\CC/\Lambda^*$, we introduce
\[(k,u)\sim (k+p, \tau(p)u),\  \tau(p)u(z)=e^{i\langle z,p\rangle} u(z),\]
for $p \in \Lambda^*.$ This way $\mathcal E_{\CC}:=\Tilde{E}/{\sim} \to \CC/\Lambda^*$ is a holomorphic vector bundle. At simple magic angles $\rank_{\CC}(\mathcal E_{\CC})=1$, while at two-fold degenerate magic angles $\rank_{\CC}(\mathcal E_{\CC})=2.$ 

We now define the complex Bloch bundle over $\CC/\Lambda^*$ (corresponding to the flat band):
\[ \begin{split}\mathcal F &:=\{(k,\phi):(\CC \times L^2_0(\CC/\Lambda;\CC^4)/{\sim} :\phi \in \indic_{0}(H_k(\alpha))\},\quad 
(k,\phi)\sim(k+p,\tau(p)\phi) \text{ for } p\in \Lambda^*.
\end{split}\]
and the corresponding real Bloch bundle $\mathcal F_0:=\{\varphi \in \mathcal F: I\varphi=\varphi\}.$ Note that such Bloch bundles corresponding to flat bands can be defined using spectral projections due to the existence of band gaps at magic angles (see \cite{bhz2,bhz3}).

To identify the complex bundle $\mathcal E_{\CC}$ with an oriented real bundle of twice the rank, we take some basis of every fiber $u_1$ (for simple magic angles) or $u_1,u_2$ (for degenerate magic angles). Focusing now exclusively on two-fold degenerate magic angles, to streamline the presentation, $u_1,iu_1,u_2,iu_2$ then defines an oriented basis. This one is always consistently oriented by general concepts \cite[Lemma $14.1$]{MS74}. To obtain an oriented basis of the bundle $\mathcal F_0$ associated with the $I$-symmetry that commutes with the Hamiltonian, we just use the symmetry $Qu(z)=\overline{u(-z)}$ and define
\[ (u_1,Qu_1),(iu_1,-iQ(u_1)), (u_2,Q(u_2)), (iu_2,-iQ(u_2)).\]

Since $\mathcal F_0 \cong \mathcal E_{\mathbb C}$ and the Euler class equals the top Chern class, we conclude that $w_1(\mathcal F_{0})=0$ and $e(\mathcal{F}_0)=c_1(\mathcal{E}_{\CC})$. Thus the Euler number of $\mathcal{F}_0$ can be computed by
\[ \chi(\mathcal F_0)=\int_{\mathbb R^2/\Gamma^*} e(\mathcal F_{0}) =\int_{\mathbb R^2/\Gamma^*} c_1(\mathcal E_{\mathbb C}).\]
\subsection{TBG with two flat bands}
At a simple magic angle, we take the real bundle $ \mathcal F_{0}$. We may add a line bundle and define: $\mathcal F_{0}' := \mathcal F_{0}\oplus L.$ 

By Proposition \ref{prop:rankg3}, we can write $\mathcal F_{0}' = \bigoplus_{i=0}^2 L_i$. As the real Bloch bundle $\mathcal F_0$ is orientable, i.e.~$w_1(\mathcal F_0)=0$, we have
\begin{equation}
w_1(L) = w_1(\mathcal F_{0}') = \sum_{i=0}^2 w_1(L_i),  
\end{equation}
as well as 
\begin{equation}
    \begin{split}
w_2(\mathcal F_{0}) = w_2(\mathcal F'_{0}) = \sum_{i<j} w_1(L_i) \smile w_1(L_j).
    \end{split}
\end{equation}

Recall that for a line bundle $L_i$, the rescaled Wannier center $\mathfrak{c}_i\in \Gamma$ is the same as the first Stiefel--Whitney class $w_1(L_i)\in H^1(\RR^2/\Gamma^*; \ZZ/2) = \Gamma/2\Gamma$. Assume $w_1(L)= \mathfrak c\in \Gamma/(2\Gamma)$, then the first condition yields that
\[ \sum_{i=0}^2 \mathfrak c_{i}= \mathfrak c . \]

Let $\omega$ be the generator of $H^2(\mathbb R^2/\Gamma^* ; \ZZ) \cong \mathbb Z$, then the Euler class of $\mathcal F_0$ is given by $e(\mathcal F_0)=-\omega,$ as the Chern number of the associated complex line bundle $\mathcal{E}_{\CC}$ is $-1$. Thus, by equation \eqref{eq:mod2}, the second Stiefel--Whitney class is given by $w_2(\mathcal F_0)=e_1\smile e_2$ with $e_1\smile e_2$ being the generator of $H^2(\mathbb R^2/\Gamma^*,\ZZ/2) \cong \ZZ/2$ as in the proof of Proposition \ref{prop:rankg3}. By the proof of Proposition \ref{prop:rankg3}, we obtain

\begin{corr}
If $L$ is a real line bundle from an isolated band, then we may decompose the real Bloch bundle $\mathcal F_0' =\mathcal F_0 \oplus L$ into a direct sum of three real line bundles $\mathcal F_0' = \bigoplus_{i=0}^2 L_{i}$ with Wannier centers
\[ \mathfrak c_0 = e_1+\mathfrak{c},\ \mathfrak c_1 = e_2+\mathfrak{c}, \text{ and } \mathfrak c_2 = e_1+e_2+\mathfrak{c}.\].
    
\end{corr}

\subsection{TBG with four flat bands}
In case of TBG at a two-fold degenerate magic angle, the real vector bundle $\mathcal F_{0}$ can be written as a direct sum of real line bundles $\mathcal F_{0} = \bigoplus_{i=0}^3 L_i.$ Using the orientability of $\mathcal F_0$, we find that the first Stiefel--Whitney class satisfies
\[\sum_{i=0}^3 w_1(L_i) =w_1(\mathcal F_{0}) =0. \]
As the Chern number of the two-fold degenerate flat band is $-1$ (cf.~\cite[Theorem~5]{bhz3}), we also obtain
\[\sum_{i<j} w_1(L_i) \smile w_1(L_j) = w_2(\mathcal F_{0}) = e_1\smile e_2.\]
Splitting a trivial bundle off the rank four real Bloch bundle $\mathcal F_{0}$ and using the proof of Proposition \ref{prop:rankg3}, we obtain the following
\begin{corr}
We can decompose $\mathcal F_0$ into a direct sum of four real line bundles $\mathcal F_0 = \bigoplus_{i=0}^3 L_{i}$ with Wannier centers at
\[
\mathfrak c_0 = e_1,\   \mathfrak c_1 = e_2,\ \mathfrak c_2 = e_1+e_2, \text{ and }\mathfrak c_3=0.
\]
    
\end{corr}


\end{document}